\documentclass[12pt]{article}
\usepackage{amsmath, amsfonts, amssymb,bm}
\usepackage{graphicx}
\usepackage{enumerate}
\usepackage{enumitem}
\usepackage{natbib}
\usepackage{url} 
\usepackage{mathtools}
\usepackage{xcolor}
\usepackage{tabularx}
\usepackage{float}
\usepackage{amsthm}
\usepackage{booktabs}
\usepackage{caption}

\newtheorem{theorem}{Theorem}[section]
{\theoremstyle{plain}

 }
\newtheorem{prop}{Proposition}

\newcommand{\bs}{\boldsymbol}

\usepackage[normalem]{ulem}

\usepackage{xr}
\newcommand{\blind}{1}

\date{}
\begin{document}

\def\spacingset#1{\renewcommand{\baselinestretch}%
{#1}\small\normalsize} \spacingset{1}


\if1\blind
{
  \title{\bf Model-based clustering of categorical data based on the Hamming distance}
  \author{Raffaele Argiento\thanks{We would like to thank the Editor, the Associate Editor, and the Referee for the insightful and
constructive comments. We also thank Federico Castelletti and Brendan Murphy for the helpful discussions. The first author's research was partially supported by MUR-Prin 2022 - Grant no. 2022CLTYP4 and the third author's research was partially supported by MUR-PRIN grant 2022 SMNNKY CUP J53D23003870008.}
    \hspace{.2cm}\\
    Department of Economics, Universit\`a degli studi di Bergamo, Italy\\
    and \\
    Edoardo Filippi-Mazzola \\
    Institute of Computing, Universit\`a della Svizzera italiana, Switzerland\\
		and\\
		Lucia Paci\\
		Department of Statistical Sciences, Universit\`a Cattolica del Sacro Cuore, Italy\\
		}
  \maketitle
} \fi

\if0\blind
{
  \bigskip
  \bigskip
  \bigskip
  \begin{center}
    {\Large\bf   Model-based clustering of categorical data based on the Hamming distance }
\end{center}
  \medskip
} \fi

\bigskip
\begin{abstract}
A model-based approach is developed for clustering categorical data with no natural ordering. The proposed method exploits the Hamming distance to define a family of probability mass functions to model the data. The elements of this family are then considered as kernels of a finite mixture model with an unknown number of components.
 Conjugate Bayesian inference has been derived for the parameters of the Hamming distribution model. The mixture is framed in a Bayesian nonparametric setting, and a transdimensional blocked Gibbs sampler is developed to provide 
full Bayesian inference on the number of clusters, their structure, and the group-specific parameters, facilitating the computation with respect to customary reversible jump algorithms. 
The proposed model encompasses a parsimonious latent class model as a special case when the number of components is fixed. 
Model performances are assessed via a simulation study and reference datasets, showing improvements in clustering recovery over existing approaches. 
\end{abstract}

\noindent%
\vfill

\newpage
 \spacingset{1.5} 
\section{Introduction}
\label{intro}
The goal of clustering analysis is to discover the underlying structure of the data and classify the observations into different subsets (or clusters) so that pairwise dissimilarities between observations
belonging to the same cluster tend to be smaller than observations in different clusters. Although many existing algorithms for clustering numerical data exist,  a limited number of methods have been specifically proposed for clustering categorical data. In this work, we are concerned with multivariate nominal data, i.e., a set of $p$ variables with a measurement scale consisting of categories with no natural ordering. 
Data on categorical scales are routinely collected in a wide range of applications, including social sciences, medical studies, epidemiology, ecology, and education \citep{agresti_book}. 

One possible avenue for clustering categorical data is based on heuristic methods that leverage a distance metric between data points. Popular examples are the K-modes algorithm \citep{Huang} and the Hamming distance-vector algorithm \citep{zhang_clustering_2006}. The former is a modified version of the K-means algorithm \citep{macqueen1967} that uses a simple matching dissimilarity measure between the modes of the clusters, while the latter finds clustering patterns by measuring the data proximity through the Hamming distance. The main drawback of distance-based methods is that they fail to quantify the uncertainty in the clustering structure, i.e., the uncertainty about any observation's group membership.

Rather, probabilistic or model-based clustering approaches allow quantifying the uncertainty associated with the clustering estimation since a statistical model is postulated for the population from which the data are sampled. 
The natural way to handle model-based clustering is through mixture models, where observations are assumed to come from one of the $M$ possible (finite or infinite) groups. Each group is suitably modeled by a density referred to as a component of the mixture and is weighted by the relative frequency (weight) of the group in the population.

 In this framework, Latent Class Models (LCMs; \citealt{goodman1974, handbook,frutta2019}) serve as the workhorse of clustering methods for categorical data. LCMs assume that the data are generated by a mixture of multivariate multinomial distributions, where each mixture
component represents a latent class (i.e., a cluster), and the variables are conditionally independent, knowing the clusters. 
LCMs are very flexible and find applications in many research fields;  see \citet{hagenaars2002} for a collection of examples. 
However, in practice, the high number of model parameters often requires imposing restrictions on the parameter space to make the inference computationally feasible. Moreover, a traditional challenge for LCMs concerns the choice of the number of classes.
A common approach to address this issue is based on model selection information criteria, which requires fitting several mixture models with an increasing number of classes. 

	In the last decades, mixture models have been intensely investigated under the  Bayesian approach.
	Arguably, one of the most important examples
	is the Dirichlet process mixture model \citep{lo1984class}, which considers 
	an infinite number of components with mixture weights obtained by a stick-breaking representation \citep{sethuraman1994constructive}.
	One of the main reasons that have made the Dirichlet process mixture very popular is the existence of relatively simple and flexible algorithms for posterior computation. Such algorithms are based on the availability of many theoretical properties of this class of models,
	 including  the random
	 discrete measure formulation,
	 the exchangeable partition distribution \citep{pitman1995} and the Blackwell–MacQueen
	urn process \citep[i.e., the Chinese restaurant process;][]{pitman1996}.
	Recently, \citet{Argiento2022IsIT} introduced a class of finite-dimensional mixture models (i.e., with a finite and possible random $M$) 
	enjoying the aforementioned theoretical properties and encompassing the finite-dimensional Dirichlet process mixture model as a special case.
	Specifically, the latter assumes a random number of mixing components and a symmetric Dirichlet prior to the weights; the model is also referred to as a mixture of finite mixtures \citep{miller2018,frutta2021}.
	On a computational side, posterior inference for finite mixture models with random $M$
	needs a transdimensional algorithm 
	that accommodates jumps between parameter spaces of different dimensions (according to the number of mixing components). 	
The Reversible Jump Markov chain Monte Carlo  
method \citep{green1995} has been the first relevant solution to deal with transdimensionality.
	The algorithm is quite popular, but it requires the design of good reversible jump moves, which poses challenges in applications, particularly in high-dimensional parameter spaces. 

On the other hand, in the Bayesian nonparametric literature, many algorithms have been proposed for posterior inference of infinite mixture models. 
These algorithms can be classified into two main groups: 
(i) \emph{conditional algorithms}, that provide full Bayesian inference on both the mixing parameters and the clustering structure 
\citep{papaspiliopoulos2008retrospective,kalli2011slice}; 
(ii) \emph{marginal algorithms}, that simplify the computation by integrating out mixture parameters and providing inference just on the clustering structure \citep{maceachern1998estimating}.  
 Exploiting the link between finite and infinite mixture models, 
  algorithms developed for Bayesian nonparametric models can be adapted to finite mixture models. 
 Examples are the Chinese restaurant process sampler in \citet{miller2018}, the telescoping sampling developed by \citet{frutta2021}, and the two augmented  Gibbs samplers proposed by \citet{Argiento2022IsIT}. 

The contribution of this work is to propose a mixture model for clustering unordered categorical data based on the Hamming distance. We first introduce a family of probability mass functions, 
built on the Hamming distance, to describe random vectors with support on a categorical space; we referred to the elements of this family as the Hamming distributions. Such distributions are characterized by two parameters that represent the distribution's center and dispersion, respectively. Conjugate Bayesian inference is derived for the two model parameters, together with a closed analytic form of the marginal likelihood. To capture the dependence between the categorical variables and simultaneously cluster the data, 
the Hamming distributions are used as kernels of a mixture model with a random number of components. We refer to the model as the Hamming mixture model (HMM). 
 Following \citet{Argiento2022IsIT}, we framed the HMM in a Bayesian nonparametric setting by assigning a prior distribution on the mixing weights through a normalization of Gamma random variables. A conditional Markov chain Monte Carlo  (MCMC) algorithm is designed to provide full posterior inference on the number of clusters, their structure, and the group-specific parameters. With all the full conditional distributions available in a closed analytical form,  
the algorithm turns to be a blocked Gibbs sampler, where transdimensional moves are automatically implied by the model. Therefore, the proposed sampling strategy offers a convenient alternative to the intensive reversible jump MCMC.

The proposed model offers a new and convenient parametrization of the LCM based on the Hamming distance, with an easy interpretation of the parameters and without restrictions on the parameter space while keeping the parsimony. Moreover, in contrast to customary LCMs that consider the number of latent classes as a fixed quantity, a key feature of the proposed model is that the number of mixing components is assumed to be random. 
The joint prior of the number of mixing components and the mixing weights induce a prior on the number of clusters, 
i.e., the components that have actually generated the data \citep{nobile2004}.  
Leveraging the connection between the LCM and the HMM, a sufficient condition for model identifiability is proved, together with the convergence of the posterior distribution of the number of components to a point mass at the true value.
Hence, the model developed in the paper extends the LCMs to the case of a random number of clusters that is automatically learned from the data, overcoming traditional issues of model selection based on information criteria. 
Moreover, the model can be used for multiple imputations of missing categorical data. For these reasons, the HHM lends itself to several potential applications.
An extensive simulation study and a set of real data examples illustrate the modeling and computational advantages of the proposed clustering approach.

The remainder of the paper is organized as follows. In Section \ref{sec:methods}, we introduce the Hamming distribution and discuss its properties, including the conjugate Bayesian inference of its parameters. Mixture modeling based on the Hamming distribution is developed in Section \ref{sec:mixture}, together with a discussion of the connections to the latent class models, the identifiability and consistency results, and the fitting details. 
Section \ref{sec:analysis} illustrates the proposed mixture model with applications to real-world datasets. Section \ref{sec:discussion} concludes the paper, and the Supplementary materials complement it with additional theoretical results and proofs, a detailed description of the MCMC sampling strategy, an extensive simulation study, and further real data examples. 

\section{Background and methods}
\label{sec:methods}

\subsection{Categorical sample space}
\label{sec:2}

Let $\bs{X}=(X_1, \dots, X_p)$ be  a vector of $p$ nominal categorical variables, or {\it attributes}, where 
each  variable $j$, for $j=1,...,p$,  can assume  $m_j$ possible levels called {\it modalities} (or categories) over the finite set $A_j=\{a_{j1},\dots, a_{jh},\dots,a_{jm_j}\}$. We denote by $\bs{x}=(x_{1},\dots,x_{p})$ a vector of  observed modalities. 
The categorical sample space is then defined as 
a collection of all possible $p$-dimensional vectors of modalities,
namely $\Omega_p= A_1 \times A_2 \times...\times A_p$, 
or equivalently,
\begin{equation*}
\Omega_p=\left\{\bs{x}=(x_1,\dots,x_p); x_1\in A_1,\dots,x_p\in A_p \right\}.
\end{equation*}
The set $\Omega_p$ is a discrete sample space of size $N=\prod_{j=1}^{p}m_j$. The cardinality of $\Omega_p$ will increase drastically as the
number of attributes and categories increases. This number can be quite large in real applications when there are many variables and categories. This can be a severe issue when performing inference based on data collected over this space. Nevertheless,  endowed with the Hamming distance $\Omega_p$  is a metric space, so some analytical properties can be used in support of the inference. 

 Given two observed vectors $\bs{x}_i$ and $\bs{x}_i^\prime$ in $\Omega_p$, the  Hamming distance \citep{hamming} 
is the number of attributes whose modalities are different in the two vectors analytically  
 \begin{equation}
 	d_H(\bs{x}_i,\bs{x}_{i^\prime})=\sum_{j=1}^p \left[1-\delta_{x_{{i}j}}(x_{i^\prime j})\right],
\label{eq:hamming} 
\end{equation}
 where  $\delta_{x_{i^\prime j}}(x_{ij})$ denotes the Kronecker delta of $x_{i^\prime j}$ and $x_{ij}$, i.e., 
\begin{equation*}
	\delta_{x_{i j}}\left(x_{i^\prime j}\right)=
	\begin{cases}
		{1 \; \text{if} \; x_{i^\prime j}\; = \; x_{ij}}\\
		{0 \; \text{if} \; x_{i^\prime j}\; \neq \; x_{ij}}.
	\end{cases}
\end{equation*}
The distance in Equation \eqref{eq:hamming} is a proper metric distance function (see \citealt{zhang_clustering_2006} for details). Thus, $\Omega_p$ endowed with the Hamming distance is a proper metric space that is also referred to as the Hamming sample space. To clarify, given any two points in $\Omega_p$, their Hamming distance can only assume a finite number of integer values ranging from $0$ to $p$. 
  In other words, if  $d_H(\bs{x}_i,\bs{x}_{i^\prime})=q$, then the two observations  differ by $q$ attributes or, equivalently, they coincide by $p-q$ attributes. 

\subsection{Hamming distribution}
\label{sec:hamming_dist}
We define a parametric family of probability mass function (p.m.f) with support $\Omega_p$ as follows. Let $\bs{c} =(c_1,\dots,c_p) \in \Omega_p$, and $\bs{\sigma}=(\sigma_1, \dots, \sigma_p)$, where $\sigma_j >0$, $j=1, \dots, p$.  
We say that a random vector $\bs{X}=(X_1,\dots,X_p)$ with support  $\Omega_p$ follows an {\it Hamming distribution} with center $\bs{c}$ and scale $\bs\sigma$ if its p.m.f. for each $\bs{x}\in\Omega_p$ is given by
\begin{eqnarray}
p(\bs{x}\mid \bs{c},\bs\sigma)=	\mathbb{P}(\bs{X}=\bs{x}\mid\bs{c},\bs\sigma)
	=\frac{1}{\prod_{j=1}^p \left(1+ \frac{m_j-1}{\exp\left({1/\sigma_j}\right)}\right)} 	\exp\left(-\sum_{j=1}^p {\frac{1-\delta_{c_j}(x_j)}{\sigma_j}}\right),
	\label{eq:pmf}
\end{eqnarray}
and we write $\bs{X}\mid \bs{c},\bs\sigma \sim\mbox{Hamming}(\bs{c},\bs{\sigma})$. 
Although a proper definition of $p(\bs{x}\mid\bs{c},\bs\sigma)$ requires to specify the support $\Omega_p$, the latter is omitted for the sake of notation. 

\begin{prop}
\label{prop:pmf}
	The function $p(\boldsymbol{x}\mid\bs{c},\bs\sigma)$ is a probability mass function on $\Omega_p$, i.e.,\\
	$\sum_{\bs{x}\in\Omega_p}p(\bs{x}\mid\bs{c},\bs\sigma)=1$.
\end{prop}

The proof of Proposition \ref{prop:pmf} is given in Section \ref{app:proof_pmf}. \\
Note that, according to the probability mass function (p.m.f.) in Equation \eqref{eq:pmf}, the $p$ categorical variables are assumed independent, i.e., $p(\bs{x}\mid \bs{c},\bs{\sigma}) = \prod_{j=1}^p p(x_j\mid c_j, \sigma_j)$, where $p(x_j\mid c_j, \sigma_j)=\left(1+ (m_j-1)\exp\left(-1/\sigma_j\right)
\right)^{-1} \exp\left(-\left(1-\delta_{c_j}(x_j)\right)/\sigma_j\right)$. In addition, when the scale parameter is constant for all attributes, i.e., $\sigma_j=\sigma>0$, $j=1, \dots, p$, then the p.m.f. simplifies to 
\begin{equation}
p(\bs{x}\mid \bs{c},\sigma)=\frac{1}{\prod_{j=1}^p \left(1+ \frac{m_j-1}{\exp\left({1/\sigma}\right)}\right)}
	\exp\left(-\frac{d_H(\bs{c},\bs{x})}{\sigma}\right),
\label{eq:pmf2}
\end{equation}
that motivates why we refer to the distribution in Equation \eqref{eq:pmf}  as the Hamming distribution. 
The p.m.f. in Equation \eqref{eq:pmf} comprises two multiplicative terms: the first term is the normalizing constant (see Appendix \ref{app:proof_pmf}),  while the second is the kernel of the distribution, i.e., $p(\bs{x}\mid \bs{c},\bs\sigma) = I(\bs{c},\bs\sigma)^{-1}g(\bs{c},\bs\sigma)$, where $I(\bs{c},\bs\sigma)=\prod_{j=1}^p \left(1+ (m_j-1)/\exp\left(1/\sigma_j\right)\right)$ and $g(\bs{x}\mid\bs{c},\bs\sigma)=\prod_{j=1}^p \exp\left(-\left(1-\delta_{c_j}(x_j)\right)/\sigma_j\right)$. Note that the kernel $g(\bs{c},\bs\sigma)$ is 
equal to one when $\bs{x}=\bs{c}$, while it is equal to $\prod_{j=1}^p \omega_j$, where $ \omega_j=\exp\left(-1/\sigma_j\right)$, when $\bs{x}\neq\bs{c}$. In other words, any attribute has a (unnormalized) weight equal to one when it coincides with the center $c_j$, while it has a (unnormalized) weight equal to $\omega_j$ when it differs from the center $c_j$. Therefore, the center $\bs{c}$ represents the unique mode of the distribution when $\sigma_j>0$, $j=1, \dots, p$. 
We notice that the p.m.f in Equation \eqref{eq:pmf} is well defined also when $\sigma_j<0$, $j=1, \dots, p$. In this case, all $\omega_j$ are greater than one, and so the parameter $\bs{c}$ is the minimum of the p.m.f.. In this work, we always assume  $\sigma_j>0$.  

As an illustration, Figure \ref{fig:pmf_plot} shows the p.m.f. for $p=2$ categorical variables that take values $A_1=\{A, B, C, D, E\}$ and $A_2=\{a,b,c,d\}$, with $m_1=5$ and $m_2=4$, respectively. The left and right panels display the p.m.f. for different values of the parameters $\bs{c}$ and $\bs\sigma$. Clearly, the highest probability is associated with the center $\bs{c}$, while $\bs\sigma$ regulates the heterogeneity of the distribution.  
\begin{figure}
	\centering
		\includegraphics[scale=0.37]{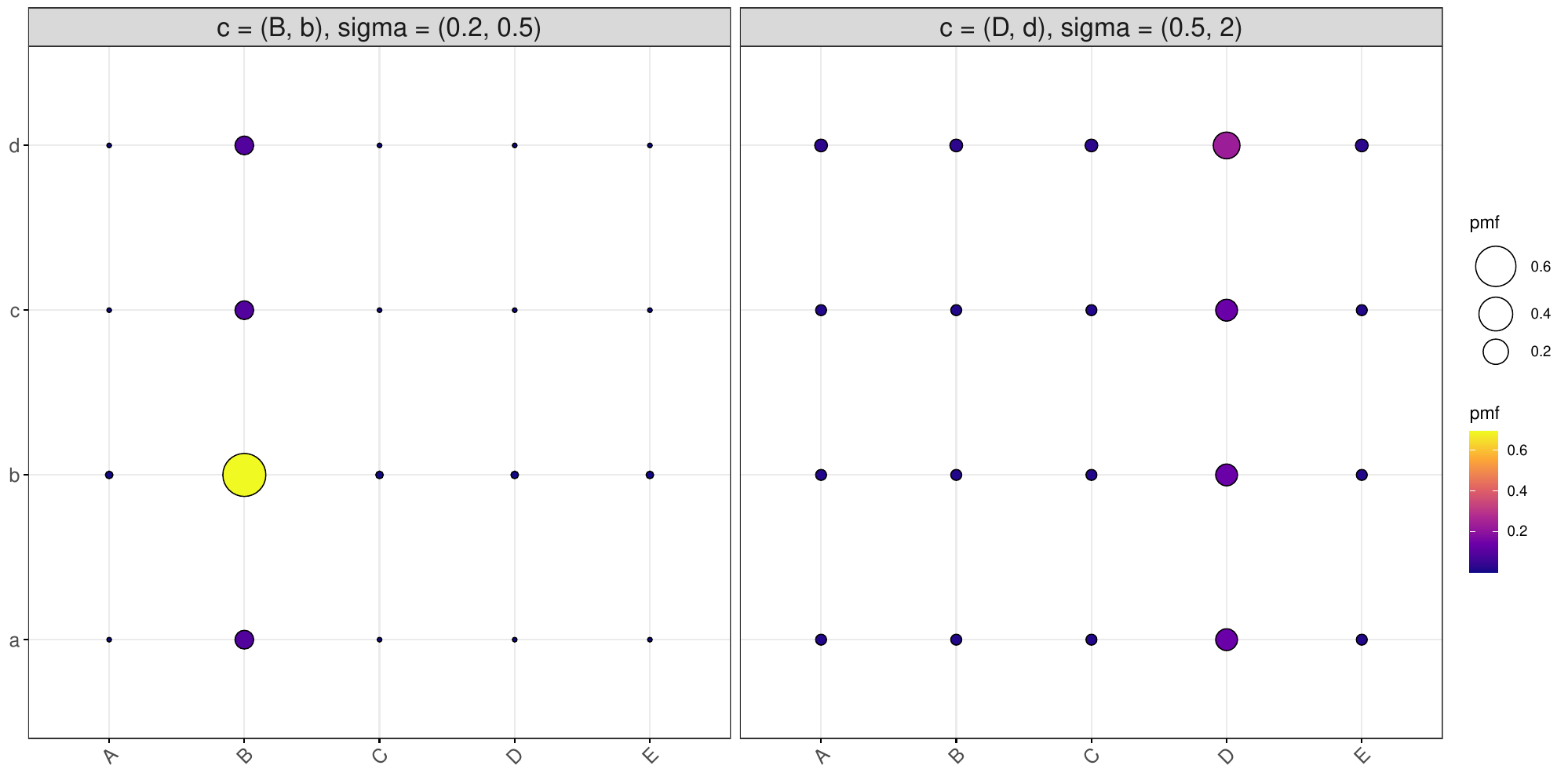}
	\caption{P.m.f. for two categorical variables with different values of $\bs{c}$ and $\bs\sigma$.}
	\label{fig:pmf_plot}
\end{figure}
%
To clarify the role of $\sigma_j$, we derive (see Appendix \ref{app:gini}) the Gini heterogeneity index \citep{gini1912} and the Shannon's entropy index \citep{shannon1948}. 
If $\bs{X}\sim\text{Hamming}(\bs{c},\bs\sigma)$, then the Gini's index of $\bs{X}$ is 
\begin{equation}
	G(\bs{X})  = 1-\sum_{\bs{x}\in\Omega_p} \left[p(\bs{x}\mid \bs{c},\bs\sigma)\right]^2
		 = 1-\prod_{j=1}^p\left[\frac{\exp\left(2/\sigma_j\right)+(m_j-1)}{\left(\exp\left(1/\sigma_j\right)+(m_j-1)\right)^2}\right].
\label{eq:Gini}
\end{equation}
Gini's index is usually normalized by dividing $G(\bs{X})$ with respect to the situation of maximum heterogeneity,  namely $G_N(\bs{X}) = G(\bs{X})/G_{\text{max}}(\bs{X})$, where $G_{\text{max}}(\bs{X})=(1-\prod_{j=1}^p 1/m_j)$.
Hence, when $\sigma_j \rightarrow 0$, for all $j$,  then $G(\bs{X})=0$ and so $G_N(\bs{X})=0$,  reflecting the condition of minimum heterogeneity. On the other side, if $\sigma_j \rightarrow \infty$, for all $j$, then $G(\bs{X})=1-\prod_{j=1}^p 1/m_j$ and so 
$G_N(\bs{X})=1$, representing the condition of maximum heterogeneity (i.e., the p.m.f. collapses to a uniform  p.m.f. on $\Omega_p$). 
In other words, $\sigma_j$ controls the heterogeneity of the p.m.f. in Equation \eqref{eq:pmf}, or equivalently, $1/\sigma_j$ plays a similar role as the index parameter in generalized linear models. 
The details the Shannon's entropy index, together with an illustration of the impact of $\sigma_j$ on the two normalized indexes, are shown in the Supplementary materials (see Appendix \ref{app:gini}).

\subsection{Bayesian inference}
\label{sec:bayes}
Let $\bs{X}_1,\dots, \bs{X}_n$ a collection of data over the discrete Hamming space $\Omega_p$.  We assume the $\bs{X}_i$'s to be conditionally independent and identically distributed, given $\bs{c}$ and $\bs{\sigma}$,  as a $\mbox{Hamming}(\bs{c},\bs{\sigma})$. Hence, given $n$ observations, the likelihood is:
\begin{equation}
p(\bs{x}_1, \dots, \bs{x}_n \mid \bs{c}, \bs{\sigma}) = \prod_{j=1}^p\left(1+\frac{m_j-1}{\exp\left(1/\sigma_j\right)}\right)^{-n}\exp\left(-\frac{n - \sum_{i=1}^n \delta_{c_j}(x_{ij})}{\sigma_j}\right).
\label{eq:like}
\end{equation}
The Bayesian model is then completed by specifying the prior distribution for the parameters. We define semi-conjugate priors for both the center and scale parameters that lead to a straightforward approximation of the posterior inference via a Gibbs sampling scheme based on the full conditional distributions presented below.

With regard to the center parameter, we consider that all $c_j\in A_j$, $j=1, \dots, p$, 
are a priori independent with a discrete uniform distribution on $\{1, \dots, m_j\}$, i.e.,  $\mbox{U}\{1,m_j\}$,  with  
\begin{equation}
p(\bs{c}) = \prod_{j=1}^p \frac{1}{m_j}. 
\label{eq:prior_c}
\end{equation}
Thus, the full conditional posterior probabilities $p(c_j \mid \text{rest})$ are proportional to\\ 
$\exp\left(-(n-\sum_{i=1}^n \delta_{c_j}(x_{ij}))/\sigma_j\right)$.

As far as the scale parameter is concerned, 
we show in the following Proposition that an independent semi-conjugate prior for $\sigma_j$, $j=1, \dots, p$, can be defined.
\begin{prop}
\label{prop:sigma}
The distribution on positive reals  with parameters $v>0$, $w>0$ and density
\begin{equation}
	f(\sigma_j\mid v,w) = \frac{w+1}{m_j^{-(v+w)}{}_2F_1\left(1,v+w;\, w+2;\, \frac{m_j-1}{m_j}\right)}\left(1+\frac{m_j-1}{\exp\left(1/\sigma_j\right)}\right)^{-(v+w)}\exp\left(-\frac{w+1}{\sigma_j}\right)\frac{1}{\sigma_j^2},
	\label{eq:prior_sigma}
\end{equation}
where $_2F_1(\cdot, \cdot; \cdot;\cdot)$ is the hypergeometric function, 
is the semi-conjugate prior for  $\sigma_j$. The updated parameters of its full conditional distribution are $v^\ast = v  + \sum_{i=1}^n \delta_{c_j}(x_{ij})$ and $w^\ast = w+ n-\sum_{i=1}^n \delta_{c_j}(x_{ij})$.
\end{prop}
The proof of Proposition \ref{prop:sigma} is given in Section \ref{app:proof_sigma}.\\
 We refer to the distribution with density given in Equation \eqref{eq:prior_sigma} as the {\it Hypergeometric Inverse Gamma} (HIG) distribution, and we write $\sigma_j\mid v,w \sim\mbox{HIG}(v,w)$.
See  \citet[Equation 9.100]{table_book} for a definition of the Hypergeometric function (also called Gauss series) and  \citet{letac2012} for a probabilistic treatment of distributions whose density involves the Hypergeometric function.

A HIG distributed random variable does not have a finite expectation for any values of the parameters $v$ and $w$. For this reason, other posterior summaries, such as the quantiles, are suggested. Instead, the corresponding unnormalized weight $\omega_j = \exp(-1/\sigma_j)$ has a finite prior mean (see Appendix \ref{app:proof_sigma}) and prior mode.
 The latter, for instance,  is equal to $w/v(m_j-1)$ when $w < v(m_j-1)$, and it is equal to one otherwise. As a result, when $w^\ast < v^\ast(m_j-1)$, the full conditional posterior  mode of $\omega_j$ can be written as 
\begin{equation*}
Mo(\omega_j\mid \mbox{rest}) = \frac{w}{v(m_j-1)}\frac{v}{v+\sum_{i=1}^n \delta_{c_j}(x_{ij})} + \frac{n-\sum_{i=1}^n \delta_{c_j}(x_{ij})}{\sum_{i=1}^n \delta_{c_j}(x_{ij})(m_j-1)}\frac{\sum_{i=1}^n \delta_{c_j}(x_{ij})}{v+\sum_{i=1}^n \delta_{c_j}(x_{ij})},
\label{eq:mode}
\end{equation*}
that is a weighted average of the prior mode of $\omega_j$ and its maximum likelihood estimate, with weights proportional to $v$ and $\sum_{i=1}^n \delta_{c_j}(x_{ij})$, respectively. Therefore, the hyperparameters $v$ and $w$ in Equation \eqref{eq:prior_sigma} can be interpreted as follows: $v+w$ represents the number of observations in a prior sample, and $w$ corresponds to the number of prior observations for which the attribute $j$ differs from its mode; equivalently, $v$ represents the number of prior observations for which the attribute $j$ is equal to its mode. 
Thus, the higher the value of $v$, the higher the weight of the prior information to the posterior inference is. 
Then, the updated parameters of the full conditional distribution clarify the additional contribution of the data, which is summarized by $\sum_{i=1}^n \delta_{c_j}(x_{ij})$. 

To sample from the HIG full conditional distribution of $\sigma_j$, we use the inversion sampling method based on the cumulative distribution function provided in Appendix \ref{app:proof_sigma} (see Equation \eqref{eq:Fsigma}). 
Sampling from the HIG distribution can also be used to choose the values of the hyperparameters of the prior distribution on $\sigma_j$. In fact, to clarify the role of the hyperparameters, we can look at the Monte Carlo distribution of the normalized Gini's index in Equation \eqref{eq:Gini} for different values of $v$ and $w$. For instance, we can elicit the absence of prior information on $\sigma_j$ by setting $v$ and $w$ that lead to a Monte Carlo distribution of the normalized Gini's index close to the uniform distribution; see Appendix \ref{app:proof_sigma} for an illustration. 

When all attributes have the same number of modalities, i.e., $m_j=m$, $j=1, \dots, p$, then a semi-conjugate analysis can be derived for a common scale parameter $\sigma_j=\sigma$, $j=1, \dots, p$. In this case, the full conditional distribution of $\sigma$ is given by: $\sigma \mid \text{rest} \sim \mbox{HIG}(v^\ast, w^\ast)$, where $v^\ast = v+np-\sum_{i=1}^n d_H(\bs{c},\bs{x}_i)$ and $w^\ast = w+\sum_{i=1}^n d_H(\bs{c},\bs{x}_i)$. Alternatively, when the number of modalities varies across the attributes, an inverse gamma prior distribution can be assumed for a common parameter $\sigma$, and a Metropolis-Hastings step can be implemented to update the scale parameter. Finally, a prior distribution on $v$ e $w$ can be placed. The natural choice is a Gamma distribution, which results in Metropolis steps for sampling from their full conditional distribution. Note that, when also $v$ and $w$ are unknown and random, the $\sigma$'s must be estimated jointly because they are not marginally independent.

As a final remark, under the semi-conjugate prior in Equation \eqref{eq:prior_sigma}, the marginal distribution of the data (i.e., the marginal likelihood) is also available in closed form for any discrete prior on $\bs{c}$ and given in the following Preposition. 
\begin{prop}
\label{prop:marginal}
Under the sampling model defined in Equation \eqref{eq:like} and the prior $\sigma_j\mid v,w \sim\mbox{HIG}(v,w)$, the marginal likelihood of the data for any discrete prior on $\bs{c}$ is
$$
p(\bs{x}_1, \dots, \bs{x}_n)= \sum_{\bs{c}\in \Omega_p} \prod_{j=1}^p p(c_j) \frac{I_j(v^\ast, w^\ast)}{I_j(v,w)},$$
	where $I_j(v,w)=\frac{m_j^{-(v+w)}}{w+1}{}_2F_1\left(1,v+w;\,w+2;\,\frac{m_j-1}{m_j}\right)$ is the normalizing constant of $f(\sigma_j\mid v,w)$, and $v^\ast$ and $w^\ast$ are the updated parameters  defined in Proposition \ref{prop:sigma}.
\end{prop}
The proof of Proposition \ref{prop:marginal} is given in Section \ref{app:post_marg}.

\section{Hamming mixture model}
\label{sec:mixture}
The sampling model in Equation \eqref{eq:like} assumes that the $p$ categorical variables are independent.
To simultaneously  capture the dependence between the variables and cluster the data, we introduce a mixture model of Hamming distributions. 
In particular, we assume that each observation $\bs{x}_i$, $i=1, \dots, n$, 
comes from a  combination of $L$ mixture components, that is
\begin{equation}
h(\bs{x}_i\mid  \bs{c}_1, \dots, \bs{c}_L, \bs{\sigma}_1, \dots, \bs{\sigma}_L, \bs\pi, L)=\sum_{l=1}^{L}\pi_l \: p(\bs{x}_i\mid\bs{c}_l,\bs{\sigma}_l),
	\label{eq:mixture}
\end{equation}
where $p(\bs{x}_i\mid\bs{c}_l,\bs{\sigma}_l)$ is the p.m.f. of the Hamming distribution defined in Equation \eqref{eq:pmf}, with center $\bs{c}_l=(c_{1l},\dots, c_{pl})$ and scale parameter $\bs\sigma_l=(\sigma_{1l}, \dots \sigma_{pl})$, with $\sigma_{jl} >0$, $j=1, \dots, p$ and $l=1, \dots, L$. The mixing weight 
$\pi_l$  is the probability that a generic observation $i$ belongs to component $l$ and it relies on the constraint $\sum_{l=1}^{L}\pi_l=1$, with $0\leq\pi_l\leq1$. 
We refer to the mixture model in Equation \eqref{eq:mixture} as the \textit{Hamming Mixture Model (HMM)}. 

Following \citet{Argiento2022IsIT}, we assign a prior distribution on the mixing weights through the normalization
of Gamma random variables. Namely, for each $l=1, \dots, L$, we let $\pi_l = S_l /T$, where $S_l$ are, conditionally on $L$, independent random variables with Gamma distribution with shape parameter $\gamma$ and unit rate, i.e., $S_1, \dots, S_L\mid L,\gamma\overset{iid}{\sim} \mbox{Gamma}(\gamma,1)$, and $T=\sum_{l=1}^L S_l$. This is equivalent to assume a symmetric $\mbox{Dirichlet}_L(\gamma, \dots, \gamma)$ prior distribution for $\bs\pi = (\pi_1, \dots, \pi_L)$, where the hyperparameter $\gamma$ regulates the prior information about the relative sizes of the mixing weights, roughly speaking - small values of $\gamma$ favor lower entropy $\pi$'s (i.e. {\it sparsity}), while large values favor higher entropy $\pi$'s \citep{rousseau2011,wall2016}. 

As a distinctive feature of the model, we assume the number of mixture components to be random with a prior distribution $q(L)$.  
To sum up,  the Bayesian HMM is given by:
\begin{equation}
\begin{array}{rll}
\bs{x}_i \mid \bs{c}_1, \dots, \bs{c}_L, \bs{\sigma}_1, \dots, \bs{\sigma}_L, S_1, \dots, S_L, L\; &\overset{iid}{\sim} \;\sum_{l=1}^{L} S_l/T \; p(\bs{x}_i\mid\bs{c}_l,\bs\sigma_l) & i=1, \dots,n \\
{S_l} \mid \gamma, L  \; & \overset{iid}{\sim} \; \mbox{Gamma}(\gamma,1) & l=1, \dots, L\\
c_{jl} \mid  L \; & \overset{iid}{\sim} \; \mbox{U}\{1, m_j\}& j=1, \dots, p;\   l=1, \dots, L\\
\sigma_{jl} \mid v, w, L \;& \overset{iid}{\sim} \; \mbox{HIG}(v, w)& j=1, \dots, p;\  l=1, \dots, L\\
L & \sim \; q(L)
\end{array}
\label{eq:model}
\end{equation}
As customary in mixture modeling, a level of hierarchy can be added to the HMM in Equation \eqref{eq:model} to facilitate the computation. To this end, we introduce a latent allocation vector $\bs{z} = (z_1\dots, z_n)$, whose element $z_i\in \{1, \dots, L\}$ denotes to which component the observation $\bs{x}_i$ is assigned; thus, $\mathbb{P}(z_i=l\mid \bs{S})=S_l/T$, where $\bs{S} = (S_1, \dots, S_L)$. In other words, we assume that the allocation variables $z_i$ are conditionally
independently distributed given $\bs{S}$ and they come
from a multinomial distribution with parameter $\bs{S}/T$, i.e., $z_i\mid \bs{S} \overset{iid}{\sim} \mbox{Multinomial}(S_1/T, \dots, S_L/T)$. The complete hierarchical specification of the HMM is given in Appendix \ref{app:model}. 
The model assumes that the categorical variables are conditionally independent given the latent allocations. This assumption is also known as the local independence assumption that basically decomposes the dependence between the variables into the mixture components, i.e., within each component, the $p$ categorical variables are independent.

The mixture model induces a clustering among the observations. To formally define the clustering, we observe that the HMM  belongs to the wide class of species sampling models, investigated in detail in \citet{pitman1996} and largely adopted in Bayesian nonparametric framework; see, among others,  \citet{ishwaran2003, miller2018, Argiento2022IsIT}.  To clarify,  we obtain ties among the latent allocations $z_1, \dots, z_n$ with a probability greater than zero. We denote $\bs{z}^\ast=(z_1^\ast, \dots, z_K^\ast)$, $K\leq L$ the unique values among these allocations. Let $\rho_n := \{C_1, \dots, C_K\}$ be the random partition (clustering) of the set $\{1, \dots, n\}$ induced by $\bs{z}^\ast$, where $C_k = \{i: z_i=z^\ast_k\}$, for $k=1, \dots, K$. In other words, two observations $\bs{x}_i$ and $\bs{x}_{i^\prime}$ belong to the same cluster $k$ if and only if $z_i = z_{i^\prime} = z_k^\ast$. 

	We highlight that the number of components  $L$ differs from the number of clusters $K$. 
	 Indeed, we notice that the values assumed by $\bs{z}^\ast$ are $K$, which is smaller or equal to $L$.
	 Then, when sampling data from the HMM,  some of the $L$ components of the mixture may turn out to be empty, i.e., no data have been generated by such components. 	This difference has been pointed out by \citet{nobile2004}, who noticed that the posterior distribution
	of the number of components $L$ might assign considerable probability to
	values greater than the number of allocated components $K$. Here, we denote $L^{(a)} = K\leq L$ the number of allocated components, namely the number of clusters, and  $L^{(na)}$  the number of non-allocated components, with $L = L^{(a)} + L^{(na)}$; the superscripts $(a)$ and $(na)$ refer to the allocated and non-allocated components, respectively. 

Both the partition $\rho_n$ and the number of clusters $K$ are random quantities whose prior distribution is induced by the model. In particular, the probability law of $\rho_n$, called exchangeable partition probability function  (eppf) in the terminology by \citet{pitman1995}, determines the (random) number of clusters $K$ and the size of each cluster $C_k$, i.e., $n_k$. 
The eppf allows us to define a generative model for $\rho_n$ that is referred to as the {\it Chinese restaurant process} in the literature. Moreover, summing over all possible values of $n_k$, the marginal prior distribution of $K$ is derived. In Appendix \ref{app:model}, we report the analytic form of the eppf and the p.m.f. of $K$; for more details, we refer the interested reader to \citet{miller2018} and \citet{Argiento2022IsIT}. 

\subsection{Connections to latent class models}
\label{subsec:LCM}
As discussed in the Introduction, the LCM is the most commonly used model for clustering categorical data. Here, we show the connection between the HMM and the LCM. The standard LCM proposed by \citet{goodman1974} assumes that each observation $\bs{x}_i$ arises from a mixture of $L$ multivariate multinomial distributions, that is
\begin{equation}
	p(\bs{x}\mid  \bs{\alpha}_1, \dots, \bs{\alpha}_L, \bs\pi, L)=\sum_{l=1}^{L}\pi_l \: p(\bs{x}\mid\bs{\alpha}_l),
\label{eq:lcm}
\end{equation}
where $\bs\pi$ is the vector of the mixing weights and $p(\bs{x}_i\mid\bs{\alpha}_l)$ is the p.m.f. of a multinomial distribution with parameter $\bs\alpha_l=(\alpha_{hjl}; h=1, \dots, m_j; j=1, \dots, p)$, i.e., $\alpha_{hjl}$ is the probability that variable $j$ has modality $h$ if observation $i$ belongs to component $l$, that is usually interpreted as a latent class.  Customary, an EM algorithm is employed to provide the maximum likelihood estimate of the $(L-1)+L \sum_{j=1}^p (m_j-1)$ model parameters, and the model is fitted with
increasing fixed values of $L$; popular information criteria are then used to select the number of classes, i.e., the number of clusters.

To reduce the number of model parameters, a parsimonious parametrization of the model in Equation \eqref{eq:lcm} has been introduced by \citet{celeux1991} and discussed in \citet{handbook}. The parsimonious model imposes a unique modal value for each variable, with all the non-modal modalities sharing uniformly the remaining mass probability; this assumption is equivalent to the one behind the HMM. Let $b_{jm}$, $j=1, \dots, p$, $l=1, \dots, L$, denote the most frequent modality of variable $j$ belonging to component $l$, i.e., $b_{jl}=\underset{h}{\arg\max}\;\alpha_{hjl}$. Then, the authors replaced the parameter $\bs\alpha_l$ in Equation \eqref{eq:lcm} by $(\bs{b}_l, \bs\varepsilon_l)$, where $\bs{b}_l = (b_{1l}, \dots, b_{pl})$ and $\bs\varepsilon_l=(\varepsilon_{1l}, \dots,  \varepsilon_{pl})$, with
\begin{equation*}
\alpha_{hjl} = \left\{\begin{array}{ll}
1-\varepsilon_{jl} & \text{if } h = b_{jl}\\
\varepsilon_{jl}/(m_j-1) & \text{otherwise}.
\end{array} \right.
\label{eq:epsilon}
\end{equation*}
When $\varepsilon_{jl}\leq (m_j-1)/m_j$,  vector $\bs{b}_l$ provides the modal levels in component $l$ for all variables, while the elements of vector $\bs\varepsilon_l$ can be regarded as scatter values. Hence, the model becomes

\begin{equation}
h(\bs{x}_i\mid  \bs{b}_1, \dots, \bs{b}_L, \bs\varepsilon_1, \dots, \bs\varepsilon_L,\bs\pi, L) = \sum_{l=1}^L\pi_l \prod_{j=1}^p (1-\varepsilon_{jl})\left(\frac{\varepsilon_{jl}}{(m_j-1)(1-\varepsilon_{jl})} \right)^{1-\delta_{b_{jl}}(x_{ij})}.
\label{eq:LCM2}
\end{equation}
First, we note that $b_{jl}$ in Equation \eqref{eq:LCM2}  corresponds to $c_{jl}$ in Equation \eqref{eq:model}. Moreover, when 
$\varepsilon_{jl} = (m_j-1)/(\exp(1/\sigma_{jl})+m_j-1)$, the mixture model in Equation \eqref{eq:LCM2} is equivalent to the HMM. 
For instance, consider the example in Section \ref{sec:hamming_dist} of $p = 2$ categorical variables with $m_1 = 5$ and $m_2 = 4$, respectively, and whose Hamming p.m.f. is shown in Figure \ref{fig:pmf_plot}. Here, $\bm{b}_l = (4,4)$ and $\bs{\varepsilon}_l = (0.36, 0.66)$, so that $\bs\alpha_l = (0.09, 0.09, 0.09, 0.64, 0.09; 0.22, 0.22, 0.22, 0.34)$, which correspond to the marginal probabilities in Figure  \ref{fig:pmf_plot}. 
In other words, the HMM provides a novel parametrization of the LCM with two main benefits: (i) an easier interpretation of the model parameters, i.e., the center and the dispersion parameters, and (ii) a random number of mixing components that leads to a full posterior inference on the number of clusters.

As a final remark, we note that the conjugate HIG prior introduced in Proposition \ref{prop:sigma} represents a two-parameter generalization of the prior for the scatter parameter used in \citet{handbook}. In fact, assuming a HIG($v,w$) prior distribution for the dispersion parameter $\sigma_{jl}$ is equivalent to a Truncated Beta distribution with parameters $v+1$ and $w-1$ for the scatter parameter $\varepsilon_{jl}$, encompassing the prior distribution employed by \citet{handbook}. 

A simplification of the model in Equation \eqref{eq:LCM2} can be considered by assuming that the scatter depends upon variables but not upon clusters and levels, i.e., $\varepsilon_{jl} = \varepsilon_l$. This model can be parametrized as a mixture of Hamming distributions with a common scale parameter $\sigma_l$; see the p.m.f. in Equation \eqref{eq:pmf2}. 
As discussed in \citet{handbook}, the LCM where the scatter does not depend on the levels suffers from a possible inconsistency when the attributes do not have the same number of modalities: the estimated probability of the modal level for a variable with few levels can be smaller than the estimated probability of the minority levels. In other words, the constraint $\varepsilon_{l}\leq (m_j-1)/m_j$ is not always satisfied. Conversely, under the HMM, such constraint is equivalent to assuming $\sigma_{l} >0$, which is always true in our setting. 

\subsection{Identifiability and consistency}
\label{sec:ident_consist}
Let 
$\mathcal{F}=\left\{p(\cdot | (\bs{c},\bs{\sigma}), \bs{c} \in \Omega_p, \bs{\sigma} \in \mathbb{R^+}^p\right\}$ be the family of Hamming distributions on the categorical sample space $\Omega$.
We follow the approach of \citet{nobile_thesis}, based on Doob's theorem \citep{doob}, to show weak convergence of the posterior distribution of the number of components $L$ to a point mass at the true value $L^\ast$. 
A crucial point to prove consistency of $L$ using Doob's theorem is that the collection of finite mixtures generated by $\mathcal{F}$ must be identifiable. In the mixture model framework, identifiability is only required up to a parameter permutation  \citep{teicher63,yakowitz68}, i.e.,  
the model must be $L!$-identifiable.  
We leverage the connections between the HMM and the LCM described in Section \ref{subsec:LCM} to provide a sufficient condition under which the HMM is $L!$-identifiable. 

\citet{allman2009} provided sufficient conditions under which the LCM is \emph{generically} identifiable, meaning that the set of non-identifiable parameters has Lebesgue measure zero. 
As detailed in Section \ref{subsec:LCM}, the HMM offers a convenient parametrization of a parsimonious LCM based on the Hamming distance. Hence, we can use the results of \citet{allman2009} to give a sufficient condition for the HMM model to be generically $L!$-identifiable. 
\begin{theorem}
\label{thm:identifiability}
The parameters of the HMM $h(\bs{x}\mid  \bs{c}_1, \dots, \bs{c}_L, \bs{\sigma}_1, \dots, \bs{\sigma}_L, \bs\pi, L)$ defined in Equation \eqref{eq:mixture} where $p \geq 3$, 
are generically $L!$-identifiable if 
\begin{equation}
\label{eq:identif}
p \geq 2 \lceil{\log L}\rceil +1,
\end{equation}
where $\lceil x\rceil$ is the smallest integer at least as large as $x$.
\end{theorem}

The proof of Theorem \ref{thm:identifiability} is reported in Section \ref{app:identif}. \\
 The result in Theorem \ref{thm:identifiability} allows us to employ Doob's theorem to prove the convergence of the posterior of the number of mixture components at the true
parameter value when the data are generated from a finite mixture over the assumed
family of Hamming distributions. 
Let $\tilde{L}$ be the maximum integer  that satisfies 
Equation \eqref{eq:identif}.
\begin{theorem}
\label{thm:consistency}
Consider $\bs{X}_1, \bs{X}_2, \dots, \bs{X}_n$  an i.i.d. sample from\\ $h(\bs{x}\mid  \bs{c}^\ast_1, \dots, \bs{c}^\ast_{L^\ast}, \bs{\sigma}^\ast_1, \dots, \bs{\sigma}^\ast_{L^\ast}, \bs\pi^\ast, L^\ast)=\sum_{l=1}^{L^\ast}\pi^\ast_l \: p(\bs{x}\mid\bs{c}^\ast_l,\bs{\sigma}^\ast_l)$, where $0\leq L^\ast \leq \tilde{L}$.\\ Assume that:
(i) the prior distribution $q(L)$ has support between one and $\tilde{L}$;
(ii) the prior density $f(\bs{\sigma})$ is any continuous function and
(iii) $\left(\bs{c}^\ast_{l},\bs{\sigma}^\ast_{l} \right) \neq \left(\bs{c}^\ast_{s},\bs{\sigma}^\ast_{s} \right)$ for any $l \neq s$. 
Then, the posterior distribution of $L$ converges almost surely to a point mass of the true number of mixture component determines, with respect to the joint prior $p(L,\bs{\pi},\bs{c},\bs{\sigma})$.
\end{theorem}
The proof of Theorem \ref{thm:consistency} is reported in Section \ref{app:identif}.\\
As discussed in \citet{Miller2023}, the limitation of a Doob-type result lies in its inability to find out whether a given true parameter value falls within the measure zero set, where consistency might not hold. Moreover, such a result is based on the assumption that the data are generated from the assumed class of finite mixture models. 
However, the posterior of the number of components in a mixture model is
especially sensitive to model misspecification \citep{miller_dunson2018, cai2021}. 
Therefore, any inferences regarding the number of components should be drawn with considerable skepticism.
On the other hand, although finite mixture models may be misspecified in practice, consistency ensures the methodology's coherence. 

\subsection{Fitting details and posterior sampling}
\label{sec:fitting}
We assume $q(L)$ in Equation \eqref{eq:model} to be the p.m.f. of a $1-$shifted Poisson distribution with parameter $\Lambda$. Hence, we set the values of hyperparameters $\gamma$ and $\Lambda$ by eliciting prior beliefs on the number of clusters $K$. 
In particular, we use the functions available in the \texttt{R} package \texttt{AntMAN} \citep{AntMAN} to find the value of $\gamma$ such that the induced prior distribution on the number of clusters is centered around a given value $K^*$, for a given value of $\Lambda$. 

The MCMC sampling strategy is based on the blocked Gibbs sampler proposed by \citet{Argiento2022IsIT}. 
An auxiliary latent variable $u$ is introduced to facilitate the computation, where $u\mid T \sim \mbox{Gamma}(n,T)$ and  $T=\sum_{l=1}^L S_l$. The benefit comes from the fact that the unnormalized weights $S_l$ are conditionally independent given the latent variable $u$ \cite{Argiento2022IsIT}.  The parameters to update are $\{u, \bs{z}, L, \bs{S}, \bs{c}_1, \dots, \bs{c}_L, \bs\sigma_1, \dots, \bs\sigma_L \}$  
and the steps of the algorithm are given in Section \ref{app:gibbs} of the Supplementary materials.

Additionally, the MCMC algorithm can naturally handle missing data. Indeed, assuming that the missing data are missing at random, a further Gibbs sampling step can be added to sample from the full conditional distribution of the missing values given all model parameters. Finally, in Section \ref{app:posterior_summary}, strategies for posterior summary are described.
The \texttt{R}/\texttt{C++} code implementing the proposed algorithm and used to carry out the analysis is available at the webpage \url{https://anonymous.4open.science/r/Hamming-mixture-model-12EA}.

\section{Empirical analysis}
\label{sec:analysis}
An extensive simulation study is carried out to assess the accuracy of the HMM in recovering a true partition. The details of the simulation design 
are deferred to the Supplementary materials (Section \ref{sec:simulation}), along with a discussion of the results. In summary, the simulation study shows that the HMM is able to recover the underlying partition and it outperforms competitors such as the K-modes \citep{Huang}, the AutoCLASS \citep{autoclass} and the HD-vector \citep{zhang_clustering_2006} algorithms,  also in the case of missing values and when data are generated from dependent variables within the clusters. 

Analysis of two real datasets are presented here: in Section \ref{sec:zoo}, we analyze a reference
dataset for which the true partition is known, while in Section \ref{sec:real_data}, we analyze a popular
high-dimensional dataset consisting of handwritten digits from the US postal services.
In the Supplementary materials, we present the analysis of two further datasets: the Soybean data, previously studied by \citet{zhang_clustering_2006}, see Section \ref{app:soy}, and the MVAD data, previously studied by \citet{mcvicar2002} and \citet{murphy2021}, see Section \ref{app:mvad}.

\subsection{Analysis of Zoo data}
\label{sec:zoo}

We illustrate the HMM on the Zoo dataset, available at the UCI machine learning repository (\url{https://archive.ics.uci.edu/})  
and previously analyzed in \citet{zhang_clustering_2006}.
The Zoo dataset consists of $101$ animals for which $p=16$ categorical features are available: hair, feathers, eggs, milk, airborne, aquatic, predator, toothed, backbone, breathes, venomous, fins, tail, domestic, cat-size, and legs. The first 15 attributes are dichotomous, while the legs variable has six categories. 
Specialists classified the animals into $7$ different classes (mammals, birds, reptiles, fish, amphibians, insects, and
mollusks), with group sizes $41$, $20$, $5$, $13$, $4$, $8$ and $10$. 

We set the hyperparameters $v_j=6$ and $w_j=0.25$ when $m_j=2$ and $v_j=3$ and $w_j=0.25$ when $m_j=6$, $j=1, \dots, p$;  this choice corresponds to a prior distribution on the normalized Gini's index that is similar to the uniform distribution, see Section \ref{sec:bayes}. 
We also set $\Lambda=7$ and $\gamma=0.68$, which leads to a prior distribution of the number of clusters $K$ centered on seven, see the left panel of Figure \ref{fig:post}. We run the  Gibbs sampler detailed in Section \ref{app:gibbs} for 25,000 iterations with a 5,000 iteration burn-in. The computational time on an Intel(R) Core(TM) i7-8550U CPU processor with a base frequency of 1.80GHz is roughly 0.01 seconds per iteration. The left panel of Figure \ref{fig:post} shows the posterior distribution of  $K$. Clearly, the data shrink the posterior probability mass function towards seven clusters. The right panel of Figure \ref{fig:post} displays the posterior similarity matrix, that is, the proportion of times that two observations have been assigned to the same mixing component over the MCMC iterations; the figure reveals the somewhat low uncertainty of the clustering structure. The clustering is then estimated using 
the variation of the information loss function described in Section \ref{app:posterior_summary}.

\begin{figure}[h]%
\centering
\includegraphics[width=0.43\textwidth]{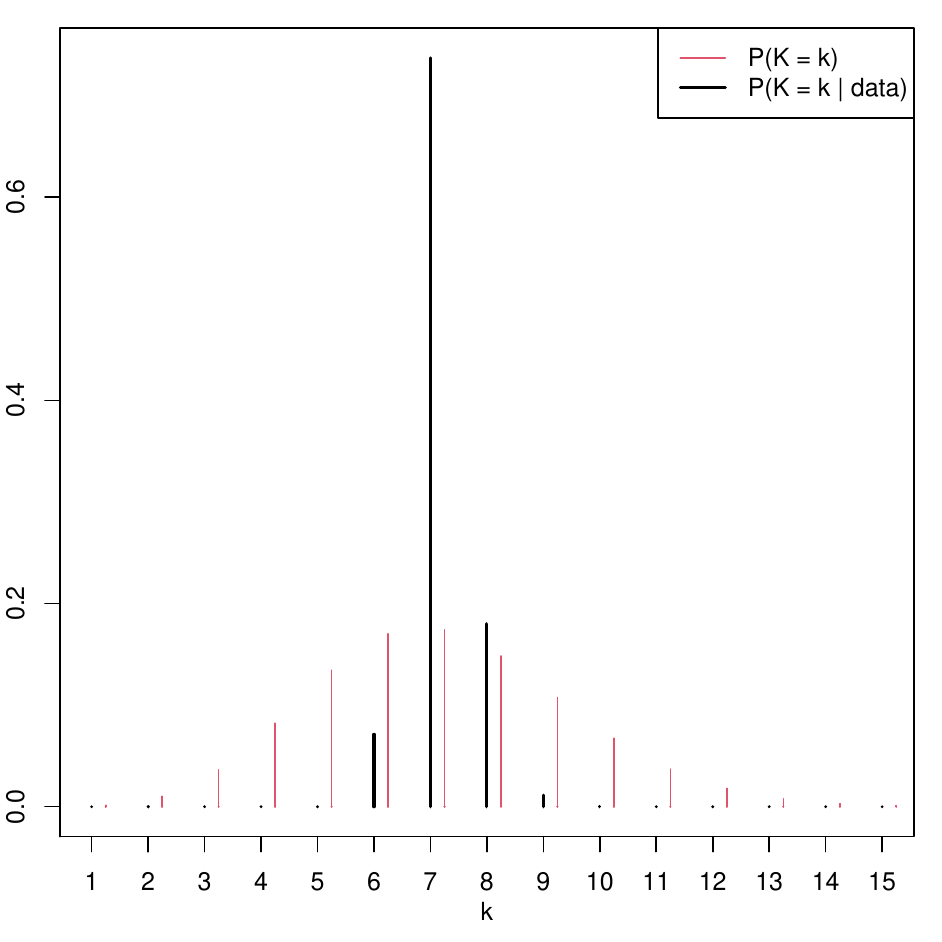}%
\includegraphics[width=0.42\textwidth]{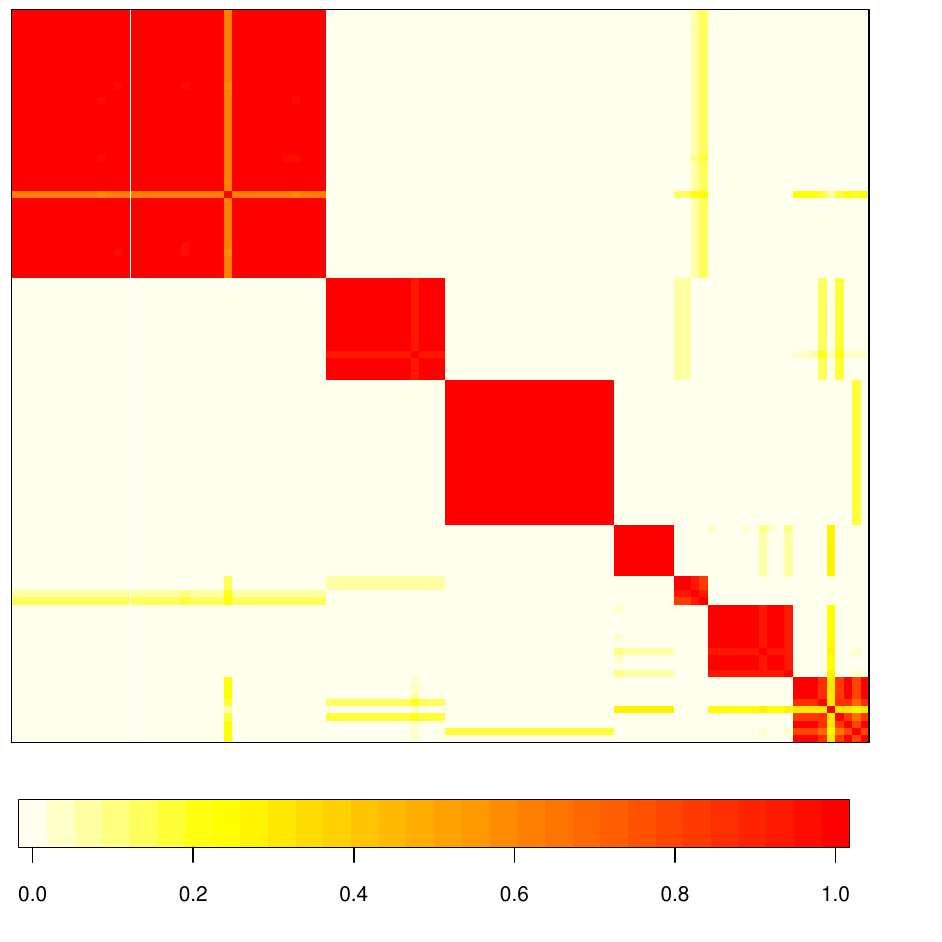}%
\caption{Prior and posterior distribution of $K$ (left panel) and posterior similarity matrix (right panel) for the Zoo dataset.}%
\label{fig:post}%
\end{figure}

Finally, we compare the clustering estimation accuracy of the HMM relative to the K-modes, the AutoCLASS  
and the HD-vector algorithms. 
The comparison is based on the adjusted Rand Index (aRI; \citealt{arandi}) computed with respect to the classification provided by the specialists; the index ranges between zero (the two clustering do not agree on any pair of points) and one (the clustering are the same). The results are presented in Table \ref{tab:comparison_zoo}. 
Because the K-modes and the HD-vector algorithms are sensitive to the initial seeds chosen to start the algorithms, they have been run 100 times with different initial values; the aRI reported in Table \ref{tab:comparison_zoo} is the mean over the 100 runs. 
Note also that the K-modes algorithm needs to specify the number
of clusters, and the results shown in Table \ref{tab:comparison_zoo} are those obtained by setting the number of clusters equal to seven. The results from the AutoClass algorithm are obtained from the Web App available at \url{https://github.com/pierrepo/autoclassweb}. The highest aRI is obtained from the HMM when the scale parameter $\sigma$ is assumed
to be shared across the 16 features (HMM$\sigma$), while the HMM with component-specific scale parameter outperforms the K-modes algorithm but not the HD-vector algorithm. 
Such results are somehow expected because the small size of some clusters challenges any mixture model-based approach, as already noted by  \citet{zhang_clustering_2006}. Hence, borrowing the strength of information across the 16 variables through a shared scale parameter, the HMM provides a better clustering estimation than the HD-vector algorithm. As a final remark, we mention that the HMM delivers more homogenous clusters in terms of Hamming distance relative to those provided by the HD-vector algorithm, as shown by the analysis based on the Silhouette index \citep {ROUSSEEUW198753} presented in Section \ref{app:zoo}.
\begin{table}%
\caption{Clustering estimation accuracy for the Zoo dataset; $\hat{K}$ is the estimated number of clusters, and aRI is the adjusted Rand index.}
\centering
\begin{tabular}{lccccc}
\toprule
& K-modes & AutoClass & HD & HMM & HMM$\sigma$\\
\midrule
$\hat{K}$ & - & 7 & 7 & 7 & 6\\
aRI & 0.70 & 0.85 & 0.93 & 0.87 & 0.95\\
\bottomrule
\end{tabular}
\label{tab:comparison_zoo}
\end{table}

\subsection{Analysis of USPS data}
\label{sec:real_data}
In many scientific fields, it is common to measure hundreds or thousands
of variables on each observation, i.e., to work with high-dimensional data. We present here an application with high-dimensional data to illustrate our model-based clustering methodology. We work with a subset of the popular USPS data from the UCI machine learning repository, consisting of handwritten digits from the US postal services. The dataset contains 1,756 images of the digits 3, 5, and 8, which are the most difficult digits to discriminate. Each digit is a 16 $\times$ 16 gray image represented as a 256-dimensional vector of pixels. We categorize the gray values, ranging from 0 to 2, in $m=6$ levels based on the following intervals: (0,0.001], (0.001,0.75], (0.75,1], (1,1.25], (1.25, 1.5], (1.5,2]. Such intervals have the same width, except for the first and last intervals containing white and black pixels, respectively. Figure \ref{fig:digits} shows a sample of the digits. 
Although the classification task may seem easy when looking at the images, the high-dimensional nature of the data makes the task more difficult for automatic methods. 
	\begin{figure}[h]
	\centering
	\includegraphics[width=0.78\textwidth]{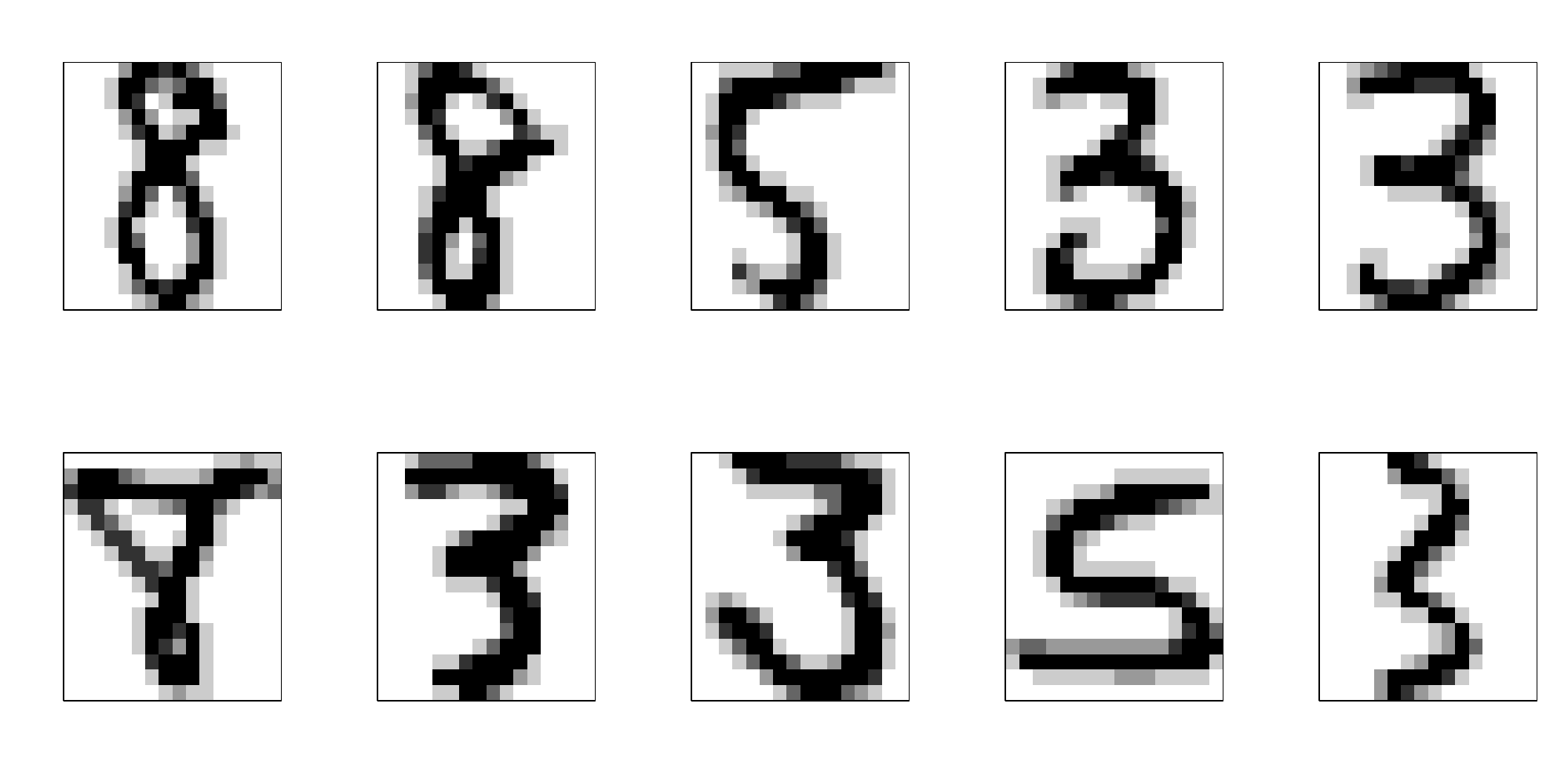}
	\vskip-0.4cm
	\caption{A sample of handwritten digits.}
	\label{fig:digits}
	\end{figure}

Again, the hyperparameters of the HIG prior are set to elicit the absence of prior information on the scale parameters, i.e., $u_j=3$ and $w_j=0.5$. We set $\Lambda=3$ and $\gamma=0.15$ so that the prior distribution on the number of clusters $K$ is centered on three. We fit the HMM running the Gibbs sampler in Section \ref{app:gibbs} for 50,000 iterations with a burn-in of 10,000 iterations. The computational cost on an Intel(R) Core(TM) i7-8550U CPU processor with a base frequency of 1.80GHz is 
roughly 0.25 seconds per iteration.%

Recall the HMM provides a new parametrization of the LCM of \citet{handbook} with a random number of components. As discussed in Section \ref{subsec:LCM}, when the number of components is fixed, a customary strategy is to select the number of classes using information criteria. Figure \ref{fig:bic_icl} displays the Bayesian information criterion (BIC; \citealt{BIC}) and the integrated completed likelihood (ICL; \citealt{biernacki2000})  derived from the EM algorithm for the LCM in Equation \eqref{eq:LCM2} as a function of the number of clusters. The results are obtained using the \texttt{R} package \texttt{Rmixmod} \citep{rmixmod}; models with smaller criteria are preferred. Clearly, for this application, the LCM of \citet{handbook} struggles to identify the number of classes. Conversely, the HMM assumes a random number of mixing components, and it is able to provide posterior inference on the number of clusters and their structure. In particular, our procedure identifies seven groups of digits based on the variation of information criterion applied to the posterior similarity matrix.  
  The posterior mode of the group-specific center parameters, obtained as discussed in Section \ref{app:posterior_summary}, are shown in Figure \ref{fig:post_cent}. For instance, the digits of the number 5 are grouped into three homogenous clusters where i) one cluster contains the digits handwritten with a long bottom tail; ii) the second cluster corresponds to the digits handwritten with almost symmetric tails; iii) the last group contains the digits handwritten with a long upper tail. A similar story holds for the digits of the number 8 that are included in two different clusters: one group contains the digits handwritten with an almost symmetric hole, while a second cluster corresponds to the digits handwritten with a wider upper hole. 

\begin{figure}[h]
\centering
\includegraphics[width=0.75\columnwidth]{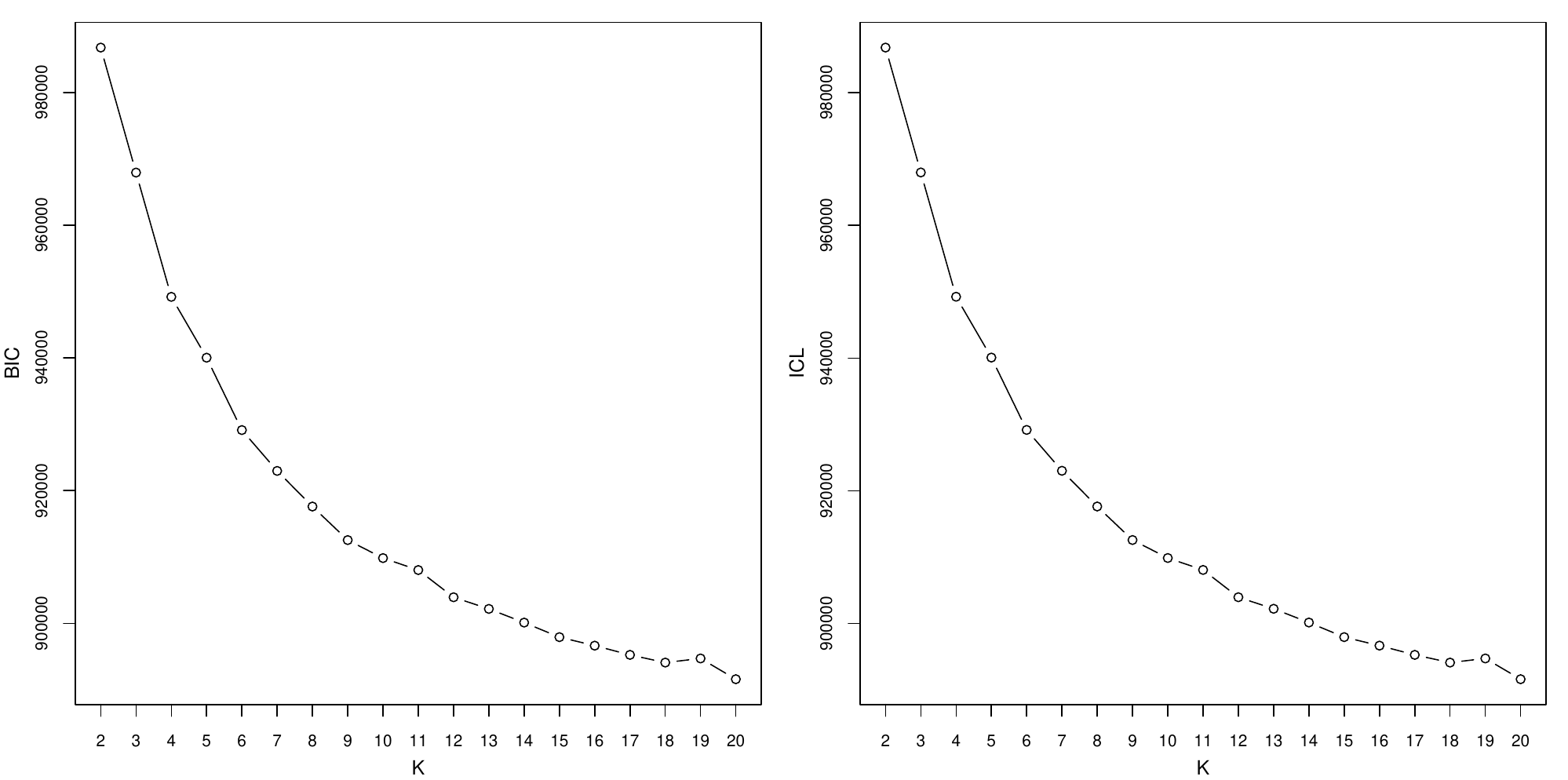}%
\caption{BIC and ICL for the LCM as a function of the number of clusters.}%
\label{fig:bic_icl}%
\end{figure}

\begin{figure}%
\centering
\includegraphics[scale=0.47]{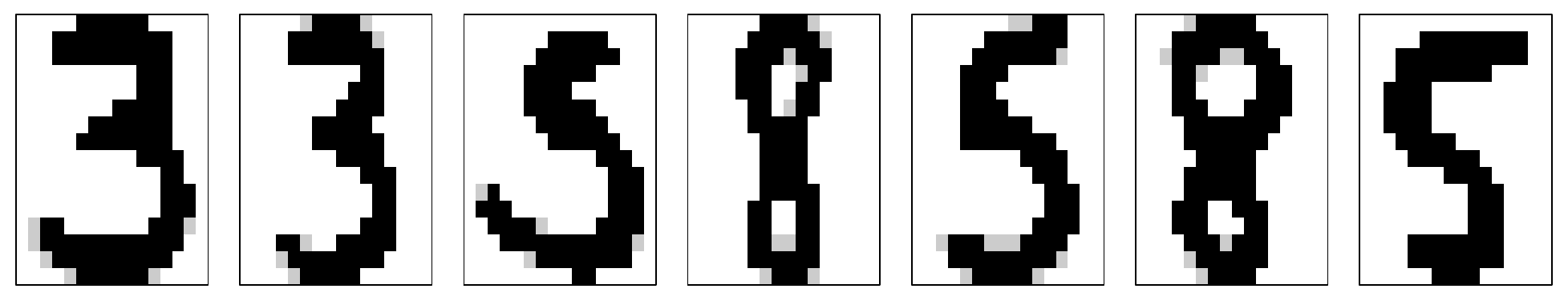}%
\caption{Posterior estimates of the group-specific center parameters.}%
\label{fig:post_cent}%
\end{figure}

\section{Discussion}
\label{sec:discussion}
A new family of parametric probability distributions built on the Hamming distance has been introduced for modeling multivariate unordered categorical data. Conjugate Bayesian inference has been derived for the parameters of the Hamming distribution model together with an analytic expression of the marginal likelihood.  To capture the dependence among the categorical variables and group the data, a finite mixture of Hamming distributions has been developed, leading to a partition of the observations into groups. We showed that the parsimonious LCM is encompassed by our specification as a special case when the number of components is fixed. Leveraging the link between the LCM and the HMM, a sufficient condition for ensuring model identifiability is established, along with the consistency of the posterior distribution of the number of the mixing components.

On the computational side, an efficient sampling strategy based on a conditional algorithm has been designed to provide full Bayesian inference, i.e., for learning the number of clusters, their structure, and the model parameters. The algorithm is a blocked Gibbs sampler with full conditional distributions available in a closed analytical form. Transdimensional moves are implied by the nonparametric formulation of the model, leading to a convenient alternative to the 
intensive reversible jump MCMC. The simulation study and the empirical results have shown the good accuracy of the novel model-based approach in recovering the clustering structure compared to existing methods. In principle, a marginal algorithm based on the Chinese restaurant representation of the finite mixture model can also be implemented since the eppf and the marginal likelihood are available in a closed analytic form, see \citet[Algorithm 2]{Neal} and \citet{Argiento2022IsIT} for details.

A limitation of the HMM is that the model cannot be straightforwardly extended to analyze mixed data, 
 i.e., when observations consist of different data types, including continuous, nominal categorical, and ordinal categorical data. Indeed, the usual model-based approach for clustering mixed data relies on the assumption of underlying continuous latent variables, see \citet{mcparland20219} for a recent review. Rather, the main idea behind the HMM is the usage of a suitable distance measure among the observable, and, to the best of our knowledge, such a distance is not available for mixed data.

Arguably, a useful next step would go beyond the local independence assumption in order to capture associations between categorical variables within each cluster. Some work has been done towards relaxing this assumption in LCMs. 
For instance, \citet{gollini2014} adopted a continuous latent variable approach 
to capture the dependence
among the observed categorical variables. Alternatively, \citet{marbac2015} assumed that the variables are
grouped into independent blocks, each one following a specific distribution that
takes into account the dependency between variables.
However, the Hamming distribution offers a natural setup for modeling an enriched dependence structure among the categorical variables within the clusters. Indeed, similar to the Mahalanobis distance obtained by rotating the Euclidean distance according to a covariance matrix, a modified Hamming distance can be obtained using a matrix capturing the dependence structure of a categorical random vector; this is the topic of our current work.

\begin{center}
{\large\bf SUPPLEMENTARY MATERIALS}
\end{center}
\begin{description}

\item[Supplementary materials] 
include: the proofs of the Propositions and Theorems, a discussion on Gini and Shannon indexes, the hierarchical specification of the model, an extensive simulation study, and additional empirical analysis. (PDF file)
\end{description}

\bibliographystyle{chicago}
\bibliography{Bibliography}

\end{document}